\documentclass[a4paper,twocolumn,superscriptaddress]{revtex4-1}
\usepackage{amsmath,amsthm,amssymb,amsfonts,mathrsfs,amsthm}
\usepackage{dsfont,yfonts,calligra}
\usepackage{graphicx,hyperref}
\usepackage{pxfonts,upgreek}
\usepackage{lmodern}
\usepackage{xcolor}
\usepackage{colordvi}

\DeclareMathAlphabet{\mathpzc}{OT1}{pzc}{m}{it}
\DeclareMathAlphabet{\mathcalligra}{T1}{calligra}{m}{n}

\def\Tr{\operatorname{Tr}} \def\>{\rangle} \def\<{\langle}
 \def\id{\mathsf{id}}
 
 \def\mE{\mathcal{E}}

\def\mN{\mathcal{N}}
\def\mL{\mathcal{L}}

\def\sH{\mathcal{H}} 
 
\def\sL{\mathsf{L}}
\def\sD{\mathsf{D}}

\def\bound{\mathsf{L}}

 \def\dec#1{\mathscr{D}}
 \def\openone{\mathds{1}}

\def\mD{\mathcal{D}}
\def\mC{\mathcal{C}}

\def\mM{\mathcal{M}}

\def\mR{\mathcal{R}}

\newcommand{\set}[1]{\mathscr{#1}}

\renewcommand{\qedsymbol}{\nobreak \ifvmode \relax \else
  \ifdim \lastskip<1.5em \hskip-\lastskip \hskip1.5em plus0em
  minus0.5em \fi \nobreak \vrule height0.75em width0.5em
  depth0.25em\fi}

\renewcommand{\ge}{\geqslant}
\renewcommand{\le}{\leqslant}

\newtheorem{corollary}{Corollary}
\newtheorem{lemma}{Lemma}
\newtheorem{proposition}{Proposition}

\newtheorem{definition}{Definition}

\theoremstyle{remark}

\theoremstyle{definition}

\newcommand{\bea}{\begin{eqnarray}}
\newcommand{\eea}{\end{eqnarray}}
\newcommand{\be}{\begin{equation}}
\newcommand{\ee}{\end{equation}}

\def\pg{P_{\operatorname{guess}}}

\begin{document}

	\title{Equivalence between divisibility and monotonic decrease of information\\in classical and quantum stochastic processes}

\author{Francesco \surname{Buscemi}}

\email{Corresponding author: buscemi@is.nagoya-u.ac.jp}

\affiliation{Graduate School of Information Science, Nagoya
	University, Chikusa-ku, Nagoya, 464-8601, Japan}

\author{Nilanjana \surname{Datta}}


\affiliation{Statistical Laboratory, University of Cambridge, Cambridge CB3 0WB, U.K.}

	\begin{abstract}
		The crucial feature of a memoryless stochastic process is that any information about its state can only decrease as the system evolves. Here we show that such a decrease of information is equivalent to the underlying stochastic evolution being divisible. The main result, which holds independently of the model of the microscopic interaction and is valid for both classical and quantum stochastic processes, relies on a quantum version of the so-called Blackwell-Sherman-Stein theorem in classical statistics.
	\end{abstract}

	\maketitle





\section{Introduction}

Discrete-time Markov chains constitute the mathematical model of \textit{memoryless stochastic processes}, i.e.,  those processes whose future state can be (statistically) predicted from the present, independently of the past; see, e.g., Refs.~\cite{parzen} and~\cite{Norris1998}. As many real-world situations appear to be approximately memoryless, Markov chains are ubiquitous in many fields, ranging from physics, chemistry, biology and information theory, to economics and social sciences.

Formally, a discrete-time stochastic process is described by an ordered sequence of random variables $(X_i)_{0\le i\le N}$, whose values $x_i$ represent the state of the system at successive time-steps $t_0\le\cdots\le t_i\le\cdots\le t_N$. The probability law of the process is given by specifying a joint $N$-point probability distribution
$\operatorname{Pr}(X_{N}=x_N;X_{N-1}=x_{N-1};\cdots;X_{0}=x_0)=p(x_N,t_N;x_{N-1},t_{N-1};\cdots;x_0,t_0)$. The process is Markovian if and only if the joint probability distribution can be factorized as follows (see, e.g., Section~6.2 in~\cite{parzen} and Theorem~1.1.1 in~\cite{Norris1998}):
\begin{equation}\label{eq:class-div}
p(x_N,t_N|x_{N-1},t_{N-1})\cdots p(x_1,t_1|x_0,t_0)p(x_0,t_0).
\end{equation}
Notice that, even though it is often assumed that the conditional probabilities $p(r,t_i|s,t_{i-1})$ in~\eqref{eq:class-div} do not depend on time (in which case the chain is called homogeneous), discrete-time Markov chains can in general be inhomogeneous, i.e., the conditional probabilities may vary with time. In any case, Eq.~\eqref{eq:class-div} suggests an insightful information-theoretic interpretation: it states that any Markov chain can always be seen, without loss of generality, as arising from an initial distribution $p(x_0,t_0)$ that propagates through an ordered sequence of independent noisy channels $(T(t_i,t_{i-1}))_{1\le i\le N}$, defined by the conditional probabilities $p(r,t_i|s,t_{i-1})$. In particular, for any pair of time-steps $(t_j, t_i)$, with $t_j\ge t_i$, there exists a noisy channel $T(t_j,t_i)$, with transition matrix $p(r,t_j|s,t_i)$, such that $p(x_j,t_j)=\sum_{x_i}p(x_j,t_j|x_i,t_i)p(x_i,t_i)$, independent of the choice of the initial distribution $p(x_0,t_0)$. This is called the \textit{divisibility property} of Markov chains~\cite{parzen}.

Such an information-theoretic description not only imparts an 
operational significance to the
memoryless property of a Markov chain, 
but also suggests ways to `quantify' it. This is done in terms of various inequalities, which show how certain information-theoretic quantities cannot increase along a Markov chain. Such inequalities, called data-processing inequalities~\cite{cover,cover2012elements}, formalize the idea that the \textit{a priori} knowledge that one has about the state of the system cannot increase along a Markov chain.

While the formalism of Markov chains (both homogeneous and inhomogeneous in time) is perfectly settled in classical probability theory, 
when trying to extend the same ideas to quantum theory, some ambiguities arise. Concerning this ongoing debate, we refer the interested reader to~\cite{Breuer2009,benatti,Bylicka2014,Rivas2010,Liu-nature,luo,Chruscinski2011,UshaDevi2011,UshaDevi2012,Smirne2013,Hall2014,Vacchini2011,Darrigo,Xu,ringbauer} and, in particular, to the recent comprehensive review by Rivas, Huelga and Plenio~\cite{Rivas}. The reasons for such ambiguities can be arguably traced back to two main factors. Firstly, there is a formal obstacle: quantum theory does not, in general, allow the description of quantum stochastic processes in terms of \textit{joint} $N$-point probability distributions or quantum states, so that any direct analogy with classical stochastic processes is irreparably lost. A thorough discussion about this point is outside the scope of the present contribution and we refer the interested reader to the discussions presented in Refs.~\cite{Vacchini2011,Rivas} and~\cite{timecorr1,timecorr2}. Secondly, a historical reason: traditionally, the definition of quantum Markov processes has been restricted to those processes which are homogeneous in time, with special emphasis on their semigroup structure (see, e.g., Ref.~\cite{kossakowski-semigroup,gorini,lindblad,alicki-lendi}). Indeed, the systematic study of inhomogeneous quantum stochastic processes is still at its infancy, and there exist different, possibly inequivalent, approaches to it~\cite{Rivas}.

The approach we adopt here is one that was first advocated by Breuer, Laine and Piilo in~\cite{Breuer2009}, namely, that a quantitative definition of `quantum Markovianity' may be possible in information-theoretic terms, by defining as `Markovian' those processes that never increase `information' over time. 
The terms given within quotation marks will be defined precisely in the next section. According to this approach, data-processing inequalities serve as witnesses of non-Markovianity, since the violation of any such inequality implies that the underlying stochastic process is non-Markovian~\cite{Breuer2009,Rivas2010,Chruscinski2011,UshaDevi2011,luo,Smirne2013,Hall2014,Rivas}. Indeed, while there are ambiguities in defining what `quantum Markovianity' is, general consensus does exist on what is \textit{not} Markovian, and processes that violate any sensible data-processing inequality should definitely be considered as non-Markovian. The following question then arises naturally: is it possible to assume, as starting point, the validity of some data-processing inequality, and derive, from such an assumption only, a complete algebraic characterization of all those processes that never violate such an inequality?

In what follows, we show that such a characterization is indeed possible, and that it exactly singles out so-called \textit{divisible processes} (see Eq.~(\ref{eq:divisibility}) below), thus supporting the idea that `information-theoretic Markovianity' is equivalent to `divisibility.' An important step in this same direction, though based on a completely different proof strategy, has already been made in Ref.~\cite{Chruscinski2011}. We note, however, that the analysis there relies on an extra assumption, which cannot be justified solely on an information-theoretic basis, on the nature of the underlying stochastic process: it is thus model-dependent and less general than the one proposed here. We will come back to this point later on, in Section~\ref{sec:discussion}.

The paper is structured as follows: in Section~\ref{sec:mappings} we introduce the notation, terminology, and basic definitions; in Section~\ref{sec:divisibility} we formally define the idea of divisibility and the memoryless property, and discuss their connections; in Section~\ref{sec:info-dec} we define information-decreasing evolutions via the notion of guessing probability; the two main results, equating the notion of information decrease with divisibility (and hence, with the memoryless property)  are discussed in Sections~\ref{sec:semiclassical} and~\ref{sec:quantum}, for the classical and the quantum cases, respectively; Section~\ref{sec:discussion} concludes the paper with some discussions and remarks.

\section{Quantum dynamical mappings}\label{sec:mappings}

In what follows, we only consider quantum systems defined on finite dimensional Hilbert spaces $\sH$. The definitions used here closely adhere to those given in standard textbooks~\cite{Petruccione2002,nielsen2010quantum}. We denote by $\bound(\sH)$ the set of all linear operators acting on $\sH$, and by $\sD(\sH)$ the set of all density operators (or states) $\rho\in\bound(\sH)$, with $\rho\ge 0$ and $\Tr[\rho]=1$. The identity operator in $\bound(\sH)$ is denoted by the symbol $\openone$. An \textit{ensemble} $\mE=\{p(x);\rho^x\}_{x \in \set{X}}$ is a finite family of states $\rho^x$ and their \textit{a priori} probabilities $p(x)$.
A \textit{positive-operator valued measure} (POVM) is a finite family of positive semi-definite operators $\{P^y\}_{y\in\set{Y}}$, such that $\sum_{y\in\set{Y}}P^y=\openone$. A \emph{quantum channel} is a linear, completely positive trace-preserving (CPTP) map $\mN:\bound(\sH_A)\to \bound(\sH_B)$. The identity channel from $\bound(\sH)$ to itself is denoted by $\id$.

The physical model we consider is that of a quantum system ($S$) which, at an initial time $t_0$, is put in contact with its surrounding environment ($E$) and allowed to evolve jointly with the latter through successive discrete instants in time $t_0\le t_1\le \cdots\le t_N$. We assume that the environment, at time $t_0$, is in some fixed state  (e.g., equilibrium state) $\sigma_{E}$, which is uncorrelated with the state of the system. The joint unitary evolution can be described by a (discrete) two-parameter family of unitary operators $U(j,i)\in\bound(\sH_S\otimes\sH_E)$, with $0\le i\le j\le N$, each one modeling the joint system-environment evolution from time $t_i$ to time $t_j\ge t_i$, and satisfying the composition law $U(k,i)=U(k,j)U(j,i)$, for all $t_k\ge t_j\ge t_i$. Of course, the consistency requirement $U(i,i)=\openone$, for all $i$, is understood.

Hence, if the initial state of the system is $\rho^0_{S}$, its state at time $t_i$ is given by
\begin{equation}\label{eq:model}
\rho^i_{S}=\Tr_{E}[U(i,0)\ (\rho^0_{S}\otimes\sigma_{E})\ U^\dag(i,0)].
\end{equation}
In this paper we also assume that the system's state in~(\ref{eq:model}) can be arbitrarily initialized, so that Eq.~(\ref{eq:model}) can be used to define a sequence of quantum channels from $\bound(\sH_S)$ into itself given by
\begin{equation*}
\mN(t_i,\bullet):=\Tr_E[U(i,0)\ (\bullet\otimes\sigma_E)\ U^\dag(i,0)],
\end{equation*}
describing the change of the reduced system from the initial time $t_0$ to later times $t_i\ge t_0$, and satisfying the consistency requirement $\mN(t_0,\bullet)=\id$. In what follows, for the sake of readability, we will denote each channel $\mN(t_i,\bullet)$ simply by $\mN^i$.

In fact, the Stinespring-Kraus unitary representation theorem~\cite{kraus1971general,Stinespring1955} guarantees that any sequence of quantum channels $(\mN^i)_{i\ge 0}$ with $\mN^0=\id$ can always be physically interpreted as arising from an open-system evolution similar to the one described above. It is hence possible (and preferable, whenever the underlying microscopic model is unknown) to start the analysis from an arbitrarily given family of channels, without further assumptions about the underlying interaction.

In order to keep our analysis general enough to encompass typical information-theoretic processes like encodings, decodings, noisy channels, quantum measurements, decision processes, coarse-grainings, etc, we allow the input and output systems to be associated with different Hilbert spaces, so that channels $\mN^i$ are linear CPTP maps, all with the same initial space $\bound(\sH_S)$ but with different output spaces $\bound(\sH_i)$. For later convenience, we summarize the above discussion in the following definition:

\begin{definition}[Quantum dynamical mappings]
	Given an initial quantum system $S$ and its Hilbert space $\sH_S$, a (discrete-time) \textit{dynamical mapping} of $S$ is given by a sequence of quantum channels $(\mN^i)_{i\ge 0}$ from $\bound(\sH_S)$ into $\bound(\sH_i)$, each modeling the operation mapping $S$ from the initial time $t_0$ to later times $t_i\ge t_0$, and satisfying the consistency requirement $\mN^0=\id$.
\end{definition}

Clearly, when we know that the system's Hilbert space does not change in time, we are back to the usual scenario in which $\sH_i\cong\sH_S$, for all $i$: all the results presented in this paper are still valid, without modification, in this special case too.

\section{Divisibility and the memoryless property}\label{sec:divisibility}

A crucial point to stress is that, while the joint system-environment evolution can always be divided (as a consequence of its unitarity) into successive steps, i.e., $U(j,i)=U(j,0)[U(i,0)]^\dag$, for all $j\ge i$, the same is not in general true for the reduced dynamics of the system $S$ alone. Namely, given a dynamical mapping $(\mN^i)_{i\ge 0}$, it is not in general possible to find a family of channels $\{\mL(j,i)\}_{j\ge i\ge0}$ from $\bound(\sH_i)$ into $\bound(\sH_j)$ such that
\begin{equation}\label{eq:divisibility}
\mN^j=\mL(j,i)\circ\mN^i,\ \textrm{for all }j\ge i\ge0.
\end{equation}
Whenever this is the case, we say that the dynamical mapping $(\mN^i)_{i\ge 0}$ is \textit{divisible}. Notice that, in order to show that a dynamical mapping $(\mN^i)_{i\ge 0}$ is divisible, it is sufficient to find another sequence of quantum channels $(\mC^j)_{j\ge1}$ from $\bound(\sH_{j-1})$ to $\bound(\sH_j)$, such that
\begin{equation}\label{eq:divisibility-maps-C}
\mN^{i+1}=\mC^{i+1}\circ\mN^{i},\ \textrm{for all }i\ge0,
\end{equation}
with the consistency requirement $\mC^1=\mN^1$.
The channels $\mL(j,i)$ in~(\ref{eq:divisibility}) are then given by
\[
\mL(j,j)=\id,\ \textrm{for all }j\ge0,
\]
and
\[ 
\mL(j,i)=\mC^{j}\circ\mC^{j-1}\circ\cdots\circ\mC^{i+1},\ \textrm{for all }j>i\ge0.
\]
A schematic representation of a divisible dynamical mapping is given in Fig.~\ref{fig:1}.
\begin{figure}[tb]
	\includegraphics[width=\columnwidth]{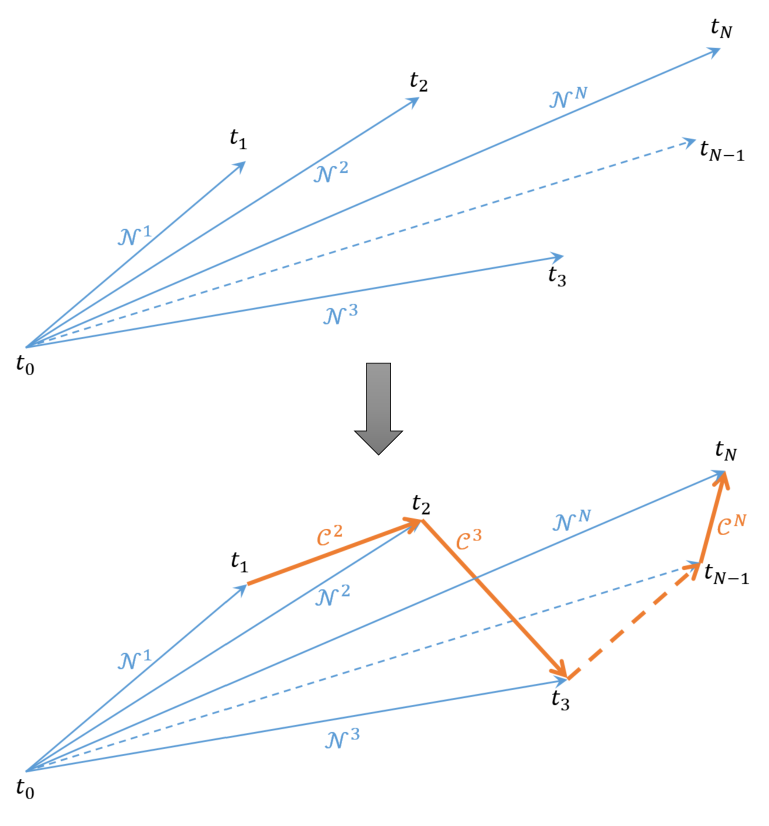}
	\caption{Schematic representation of a discrete-time dynamical mapping. \textbf{Top}: the system's evolution from an initial time $t_0$ to the final time $t_N$ is represented by a sequence of quantum channels $(\mN^i)_{0\le i\le N}$ (the thin arrows in the figure), describing the stochastic evolution of the system from the initial time $t_0$ to some later time $t_i\ge t_0$. \textbf{Bottom}: the dynamical mapping is divisible if there exists another sequence of quantum channels $(\mC^i)_{1\le i\le N}$ (the thick arrows) such that $\mN^{i+1}=\mC^{i+1}\circ\mN^{i}$, for all $0\le i\le N-1$.}
	\label{fig:1}
\end{figure}

\subsection{The memoryless property}

Divisibility, as described above, is intimately related with the memoryless property as follows. Suppose that a joint system-environment interaction (on which we do not make any particular assumption) gives rise to a reduced system dynamics that is divisible. This means that, from the point of view of an observer without direct access to the environment, the evolution of the system is completely indistinguishable from the sequential application of independent quantum channels $(C^i)_{i\ge 1}$, each modeling the evolution of the system from time-step $t_{i-1}$ to $t_i$. 

Then, the Stinespring-Kraus theorem~\cite{Stinespring1955,kraus1971general} says that a sequential application of channels $(C^i)_{i\ge 1}$ can be represented as a sequence of independent unitary interactions $(W_1,W_2,W_3,\cdots)$ of the system with a corresponding sequence of independent environments, initialized in states $(\sigma_E,\sigma^{(2)}_E,\sigma^{(3)}_E,\cdots)$ and discarded after the interaction. Such a model, similar to collision-like~\cite{collision}  or power-dilation~\cite{buscemi-power} models, is depicted in Fig.~\ref{fig:2}.

\begin{figure}[tb]
	\includegraphics[width=\columnwidth]{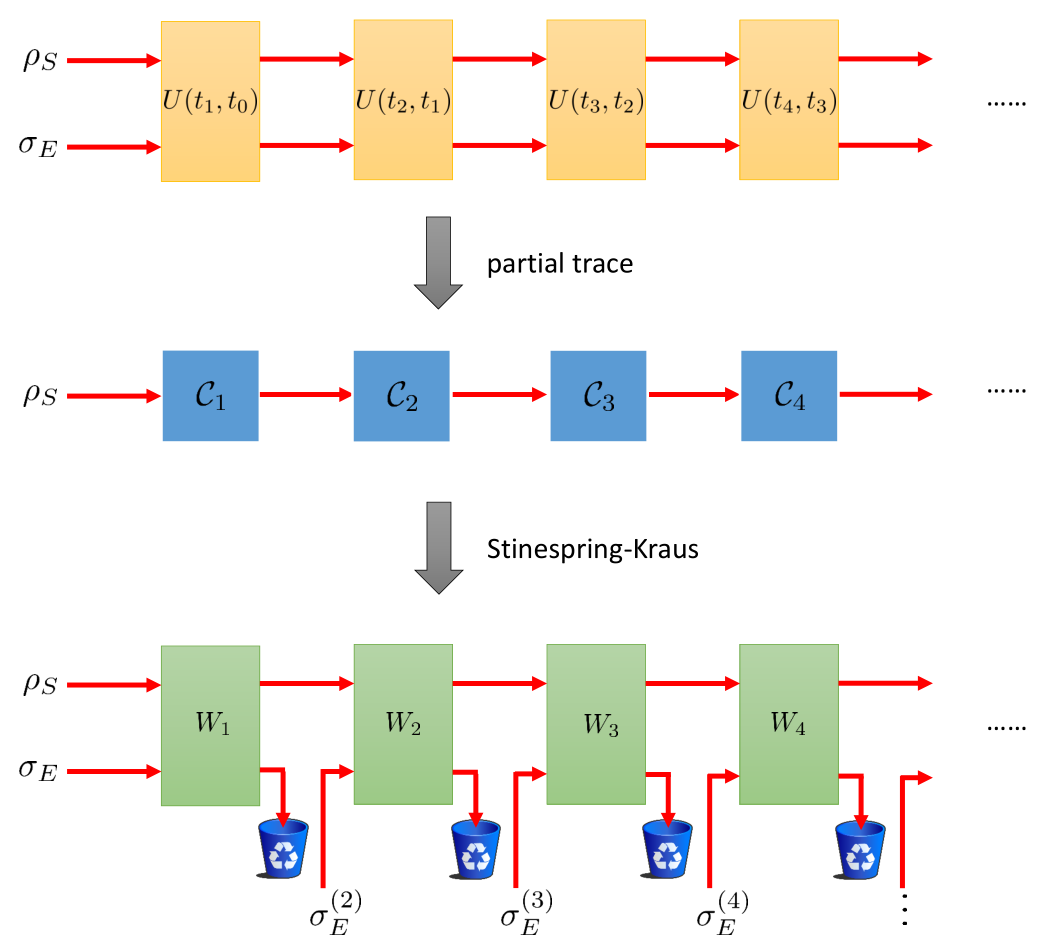}
	\caption{\textbf{Top}: the joint system-environment evolution. \textbf{Middle}: the reduced dynamics of the system is divisible if it cannot be distinguished from a sequence of independent quantum channels $(\mC^i)_{i\ge1}$ applied in series. \textbf{Bottom}: as a consequence of the Stinespring-Kraus representation theorem~\cite{Stinespring1955,kraus1971general}, any divisible process can always be thought as originating from a collision-like model, in which the system interacts with an environment that is reset after every interaction.}
	\label{fig:2}
\end{figure}

It is then clear that divisibility implies, as a consequence of the Stinespring-Kraus representation theorem, a strong form of the memoryless property: no memory can be kept of the past, because the system's evolution is indistinguishable from that arising from the interaction with an environment that is reset at each time-step.

Two remarks are in order at this point. First, the fact that a divisible quantum mapping admits a collision-like model does not mean that the underlying joint evolution \textit{actually is} collision-like. In general, it is possible that correlations between the system and its environment are established and kept along the evolution; however, if the system's dynamics is divisible, such correlations do not give rise to any observable memory effect. This is in line with the fact that a well-defined CPTP reduced dynamics exists also for strongly correlated systems~\cite{Buscemi2013}. 

The second remark is about weaker notions of divisibility, most notably the so-called P-divisibility property~\cite{Rivas}, which holds whenever there exists a sequence of (not necessarily completely) \textit{positive}  trace-preserving maps $(\mathcal{P}^i)_{i\ge 1}$ satisfying Eq.~\eqref{eq:divisibility-maps-C}. In such a case, the Stinespring-Kraus representation theorem does not hold, so that the relation with the memoryless property is lost. In addition, it is possible to \textit{observe} memory effects just by suitably extending the system (see Corollary~\ref{cor:memory-effects} below).

\section{Information-decreasing dynamical mappings}\label{sec:info-dec}

As we mentioned before, the fact that information can only decrease along a Markov chain can be formalized in many ways, via a number of data-processing inequalities~\cite{czizar,cover,Rivas}. In what follows, we focus on one such data-processing inequality, which enjoys a simple definition and a natural interpretation.

Suppose that the experimenter knows \textit{a priori} that the system's `true' state belongs to a known family of possible states $\{\rho^x_S\}_x$, and that each element $\rho^x_S$ can be the `true' state with probability $p(x)$. The experimenter's initial (partial) knowledge about the system is therefore modeled by an ensemble $\mE=\{p(x);\rho^x_S\}_x$. In this situation, therefore, the information initially possessed by the experimenter depends on the distinguishability of the states in the ensemble: the higher the distinguishability of the states $\rho^x_S$, the more information is available to the experimenter. A natural measure of the information about the system's initial state is therefore given by the \textit{guessing probability}~\cite{holevo1973statistical,yuen1975}
\[
\pg(\mE):=\max\sum_xp(x)\Tr[P^x_S\ \rho^x_S ],
\]
where the maximization is over all POVMs $\{P^x_S\}_x$ on $\sH_S$. The fact that the guessing probability cannot increase under the action of a channel on the states of the ensemble is a very simple consequence of its definition.

\begin{definition}\label{def:info-dec}
	A discrete-time dynamical mapping $(\mN^i)_{i\ge0}$ is said to be \emph{information decreasing} if and only if, for any ensemble $\mE=\{p(x);\rho^x_S\}_x$, the ordered sequence of guessing probabilities $(\pg(\mE_i))_{i\ge0}$, where \[\mE_i:=\{p(x);\mN^i(\rho_{S}^x)\}_x,\] is monotonically non-increasing, i.e., \[\pg(\mE_i)\ge\pg(\mE_{i+1})\] for all $i\ge0$.
\end{definition}

The above definition constitutes our formalization of the fact that the information about the initial state of the system does not increase as the system evolves. This, of course, has to happen \textit{irrespective} of the information about the system that the experimenter initially has. This fact is reflected in the above definition by the requirement that the guessing probability cannot increase \textit{for any ensemble of initial states}, i.e., for any finite set $\set{X}=\{x\}$, any probability distribution on $\set{X}$, and any collection of states $\rho^x_S\in\sD(\sH_S)$.

This paper builds upon a series of results extending the so-called Blackwell-Sherman-Stein theorem~\cite{blackwell1953equivalent,torgersen1970comparison,Torgersen1991} of classical statistics to quantum statistical decision theory~\cite{buscemi2005clean,shmaya2005comparison,chefles2009quantum,buscemi2012comparison,buscemi2012all,Buscemi2014}. In particular, a crucial role in this paper is played by the following result:


\begin{lemma}[\cite{Buscemi2014}]\label{lemma:ordering}
	Given two channels $\mN^i:\sL(\sH_S)\to\sL(\sH_i)$, $i=1,2$, the following are equivalent:
	\begin{enumerate}
		\item $\pg(\mE_1)\ge\pg(\mE_2)$ for any ensemble $\mE=\{p(x);\rho^x_S\}_x$, where $\mE_i:=\{p(x);\mN^i(\rho^x_S)\}_x$;
        \item $\pg(\mE_1)\ge\pg(\mE_2)$ for any ensemble $\mE=\{p(x);\rho^x_S\}_x$ with $\sum_xp(x)\rho^x_S=(\operatorname{dim}\sH_S)^{-1}\openone_S$;
		\item for any POVM $\{Q^y\}_y$ on $\sH_2$, there exists a corresponding POVM $\{P^y\}_y$ on $\sH_1$ such that $\Tr[\mN^1(\rho_S)\ P^y]=\Tr[\mN^2(\rho_S)\ Q^y]$, for any $y\in\set{Y}$ and for any $\rho_S\in\sD(\sH_S)$.
	\end{enumerate}
\end{lemma}

\begin{proof}
	That (1) implies (2) is obvious, since point (2) considers only a subset of all possible ensembles, i.e., those whose average state is the completely mixed state $(\operatorname{dim}\sH_S)^{-1}\openone_S$. Also the implication (3) $\Rightarrow$ (1) is trivial, since, if (3) holds, then any statistics obtainable from the outputs of $\mN^2$ can also be obtained from the outputs of $\mN^1$.
	
	The remaining implication, i.e., (2) $\Rightarrow$ (3), is a direct consequence of Theorem~3 of Ref.~\cite{Buscemi2014}: point number (2) above corresponds to point number (4) there -- point number (3) above corresponds to point number (3) there.
\end{proof}

In other words, the guessing probability for a channel $\mN^1$ (for any ensemble) is greater than or equal to that for a channel $\mN^2$ if and only 
if the image of $\mN^2$ (in the Heisenberg picture) is contained in that of $\mN^1$.  
We use the notation $\mN^1\succeq \mN^2$ (which denotes a partial ordering between channels) whenever one of the above conditions holds.
Accordingly, a dynamical mapping $(\mN^i)_{i\ge0}$ is information decreasing if and only if \[\mN^0\succeq\mN^1\succeq\cdots\succeq\mN^i\succeq\cdots.\]

\section{The semiclassical case}\label{sec:semiclassical}

Even though our analysis has been focused so far on the case of quantum dynamical mappings, we now show how the classical case too can also be treated within the same framework, under some further assumptions. This is done by passing through the intermediate case of `semiclassical' dynamical mappings, defined as sequences of channels $(\mN^i)_{i\ge0}$ with commuting output, i.e.,
\[
[\mN^i(\rho_S),\mN^i(\rho'_S)]=0,
\]
for all $i\ge 1$ and for all $\rho_S,\rho'_S\in\sD(\sH_S)$. Notice that the commutativity condition is required to hold for all times $t_i>t_0$: the system is assumed to be initially quantum and its state space, at $t_0$, is generally non-commuting. One can think of semiclassical dynamical mappings as sequences of completely decohering (sometimes dubbed `quantum-to-classical' or just `qc') channels, e.g, sequences of complete non-degenerate projective (von Neumann) measurements.


We can now state the first main result of this paper as follows:
\begin{proposition}[Semiclassical case]\label{prop:classical}
	A given discrete-time semiclassical dynamical mapping is divisible if and only if it is information decreasing.
\end{proposition}

The above is a direct consequence of the following lemma:

\begin{lemma}\label{lemma:semicl}
	Let $\mN:\sL(\sH_A)\to\sL(\sH_B)$ and $\mN':\sL(\sH_A)\to\sL(\sH_{B'})$ be two CPTP maps. Suppose that the output of $\mN'$ is abelian, i.e., $[\mN'(\rho),\mN'(\sigma)]=0$, for any $\rho,\sigma\in\sD(\sH_A)$.
	
	Then, $\mN\succeq\mN'$ if and only if there exists a third CPTP map $\mC:\sL(\sH_B)\to\sL(\sH_{B'})$ such that
	\[
	\mN'=\mC\circ\mN.
	\] 
\end{lemma}

\begin{proof}
	Here we prove only the `only if' part of the statement, as the `if' part is trivial. The proof is based on the analogous result for bipartite states derived in Ref.~\cite{buscemi2012comparison}.
	
	Since the outputs of $\mN'$ are all commuting, it is possible to find a basis $\{|i_{B'}\>\in\sH_{B'}\}_i$ that diagonalizes them all simultaneously. A simple identity then gives:
	\begin{equation}\label{eq:simple-id}
	\mN'(\rho_A)=\sum_i|i_{B'}\>\<i_{B'}|\Tr\{\mN'(\rho_A)\ |i_{B'}\>\<i_{B'}| \}, 
	\end{equation}
	for all $\rho_A\in\sD(\sH_A)$. On the other hand, since $\mN\succeq\mN'$, and since $\{|i_{B'}\>\<i_{B'}|\}_i$ constitutes a well-defined POVM on $\sH_{B'}$, we know that there exists a POVM $\{P^i_B\}$ on $\sH_B$ such that
	\[
	\Tr\{\mN'(\rho_A)\ |i_{B'}\>\<i_{B'}| \}= \Tr\{\mN(\rho_A)\ P^i_B \},
	\]
	for all $\rho_A\in\sD(\sH_A)$ and all $i$. Inserting the above equation into~(\ref{eq:simple-id}), one obtains the identity
	\[
	\mN'(\rho_A)=\sum_i|i_{B'}\>\<i_{B'}|\Tr\{\mN(\rho_A)\ P^i_B \},
	\]
	valid for all $\rho_A\in\sD(\sH_A)$,
	which can be equivalently written as $\mN'=\mC\circ\mN$ upon defining the CPTP map $\mC$ as
	\[
	\mC(\bullet_B):=\sum_i|i_{B'}\>\<i_{B'}|\Tr\{\bullet_B\ P^i_B \}.
	\]
	The above equation shows, in particular, that the map $\mC:\sL(\sH_B)\to\sL(\sH_{B'})$ is a CPTP map defined everywhere, as claimed in the statement.
\end{proof}

\subsection{The fully classical case}

The fully classical case corresponds to the situation, in which the evolving system is assumed to be classical already from the start (i.e., from $t_0$ included). Indeed, when the evolving system is classical, under the customary correspondence between diagonal density matrices and probability distributions, CPTP maps become conditional probabilities. In this way, the formalism of classical Markov chains can be recovered:

\begin{corollary}[Classical case]
	Let $(N^i)_{i\ge0}$ be a sequence of noisy channels represented by the conditional probabilities $p(x_i,t_i|x_0,t_0)$, modeling the evolution of an initial random variable $X_0$ to successive time-steps $t_i\ge t_0$. Then, the sequence $(N^i)_{i\ge0}$ is information decreasing if and only if, for any initial distribution $p(x_0,t_0)$ of $X_0$,
	there exists a Markov chain $(X_j)_{j\ge 0}$ whose two-point marginals $(X_i,X_0)$ are distributed according to $p(x_i,t_i;x_0,t_0)=p(x_i,t_i|x_0,t_0)p(x_0,t_0)$, for all $i\ge 1$.
\end{corollary}

In other words, a classical dynamical mapping never increases the distinguishability of any initial ensemble of probability distributions, if and only if it can always be `embedded' in an underlying Markov chain.

\section{The quantum case}\label{sec:quantum}

We now turn to the case in which the discrete-time dynamical mapping is fully quantum, i.e., the channels in the sequence $(\mN^i)_{i\ge0}$ are linear, CPTP maps with non-commuting outputs. In this case, it is customary to allow the evolving quantum system, originally associated with the Hilbert space $\sH_S$, to be part of a larger system, associated with a tensor product space $\sH_{S'}\otimes\sH_S$. Accordingly, we reformulate Definition~\ref{def:info-dec} to take into account such possible extensions:
\begin{definition}\label{def:ext-info-dec}
	A discrete-time dynamical mapping $(\mN^i)_{i\ge0}$ is said to be \emph{completely information decreasing} if and only if, for any auxiliary Hilbert space $\sH_{S'}$ and for any finite ensemble $\widetilde{\mE}=\{p(x);\rho^x_{S'S}\}_x$ of states on $\sH_{S'}\otimes\sH_S$, the ordered sequence of guessing probabilities	$(\pg(\widetilde{\mE}_i))_{i\ge 0}$, where \[\widetilde{\mE}_i:=\{p(x);(\id_{S'}\otimes\mN^i_S)(\rho_{S'S}^x)\}_x,\] is monotonically non-increasing, i.e., \[\pg(\widetilde{\mE}_i)\ge\pg(\widetilde{\mE}_{i+1})\] for all $i\ge 0$.
\end{definition}
Using the partial ordering notation $\succeq$ previously introduced, we can equivalently say that the process described by the dynamical mapping $(\mN^i)_{i\ge 0}$ is completely information decreasing if and only if \[(\id_{S'}\otimes\mN^0_S)\succeq(\id_{S'}\otimes\mN^1_S)\succeq\cdots\succeq(\id_{S'}\otimes\mN^N_S),\]
for all auxiliary systems $S'$. Then, the following statement holds:
\begin{proposition}[Quantum case]\label{prop:quantum}
	A given discrete-time quantum dynamical mapping is divisible if and only if it is completely information decreasing.
\end{proposition}

The above is a direct consequence of the following:

\begin{lemma}\label{lemma:quantum}
	Given a pair of CPTP maps $\mN:\sL(\sH_A)\to\sL(\sH_B)$ and $\mN':\sL(\sH_A)\to\sL(\sH_{B'})$, let $\sH_{B''}$ be an auxiliary Hilbert space isomorphic with $\sH_{B'}$, i.e., $\sH_{B''}\cong\sH_{B'}$, and $\id:\sL(\sH_{B''})\to\sL(\sH_{B''})$ the corresponding identity map.
	
	Then, $\id\otimes\mN\succeq\id\otimes\mN'$ if and only if there exists a third CPTP map $\mC:\sL(\sH_B)\to\sL(\sH_{B'})$ such that
	\[
	\mN'=\mC\circ\mN.
	\] 
\end{lemma}

\begin{proof}
	Here we prove only the `only if' part of the statement, as the `if' part is trivial. The proof presented here is based on a series of results that appeared in Ref.~\cite{shmaya2005comparison,chefles2009quantum,buscemi2012comparison,Buscemi2014}.
	
	By hypothesis, it holds that $\id\otimes\mN\succeq\id\otimes\mN'$, which implies, in particular, that, for any finite alphabet $\set{Y}=\{y\}$ and any POVM $\{Q^y_{B''B'}\}_y$ there exists a POVM $\{P^y_{B''B}\}_y$ such that
	\begin{equation}\label{eq:exteded-id}
	\begin{split}
	&\Tr\left[\{\omega_{B''}\otimes\mN'(\rho_A)\}\ Q^y_{B''B'}\right]\\  =& \Tr\left[\{\omega_{B''}\otimes\mN(\rho_A)\}\ P^y_{B''B}\right],
	\end{split}
	\end{equation}
	for all $y\in \set{Y}$, all $\omega_{B''}\in\sD(\sH_{B''})$, and all $\rho_A\in\sD(\sH_A)$.
	
	Upon introducing another auxiliary Hilbert space $\sH_{B'''}\cong\sH_{B''}\cong\sH_{B'}$ and a maximally entangled state $|\Phi^+_{B'''B''}\>\in\sH_{B'''}\otimes\sH_{B''}$, the condition expressed in Eq.~(\ref{eq:exteded-id}) can be rewritten as follows: for any alphabet $\set{Y}=\{y\}$ and any POVM $\{Q^y_{B''B'}\}_y$ there exists a POVM $\{P^y_{B''B}\}_y$ such that
	\begin{equation}\label{eq:indentity-reals}
	\begin{split}
	&\Tr\left[\{|\Phi^+_{B'''B''}\>\<\Phi^+_{B'''B''}|\otimes\mN'(\rho_A)\}\ \{\Omega_{B'''}\otimes Q^y_{B''B'}\}  \right]\\
	=& \Tr\left[\{|\Phi^+_{B'''B''}\>\<\Phi^+_{B'''B''}|\otimes\mN(\rho_A)\}\ \{\Omega_{B'''}\otimes P^y_{B''B}\}\right] ,
	\end{split}
	\end{equation}
	for all $y\in \set{Y}$, all $\rho_A\in\sD(\sH_A)$, and all $0\le\Omega_{B'''}\in\sL(\sH_{B'''})$.
	
	We now make use of the simple fact that, $\Tr[XA]=\Tr[YA]$ for all $A\ge 0$ if and only if $X=Y$, to reformulate condition~(\ref{eq:indentity-reals}), involving positive numbers, into a condition involving operators: for any alphabet $\set{Y}=\{y\}$ and any POVM $\{Q^y_{B''B'}\}_y$ there exists a POVM $\{P^y_{B''B}\}_y$ such that
	\begin{equation}\label{eq:indentity-op}
	\begin{split}
	&\Tr_{B''B'}\left[\{|\Phi^+_{B'''B''}\>\<\Phi^+_{B'''B''}|\otimes\mN'(\rho_A)\}\ \{\openone_{B'''}\otimes Q^y_{B''B'}\}  \right]\\&= \Tr_{B''B}\left[\{|\Phi^+_{B'''B''}\>\<\Phi^+_{B'''B''}|\otimes\mN(\rho_A)\}\ \{\openone_{B'''}\otimes P^y_{B''B}\}\right] ,
	\end{split}
	\end{equation}
	for all $y\in \set{Y}$ and all $\rho_A\in\sD(\sH_A)$.
	
	Now we recall the protocol of (generalized) teleportation of Ref.~\cite{braunstein}, according to which one can always choose the alphabet $\set{Y}=\{y\}$ and the POVM $\{Q^y_{B'B''}\}_y$ in Eq.~(\ref{eq:indentity-op}) such that
	\begin{widetext}
	\[
	\mN'(\rho_A)=\sum_y \mathcal{U}^y_{B'''}\circ\Tr_{B''B'}\left[\{|\Phi^+_{B'''B''}\>\<\Phi^+_{B'''B''}|\otimes\mN'(\rho_A)\}\ \{\openone_{B'''}\otimes Q^y_{B''B'}\}  \right],
	\]
	for all $\rho_A\in\sD(\sH_A)$,
	where the maps $\mathcal{U}^y:\sL(\sH_{B'''})\to\sL(\sH_{B'})$ are suitable unitary CPTP maps, i.e., $\mathcal{U}^y(\bullet)=U_y\bullet U_y^\dag$ with $U_y^\dag U_y=\openone_{B'}$. Then, condition~(\ref{eq:indentity-op}) guarantees the existence of a POVM $\{P^y_{B'''B}\}_y$ such that
	\[
	\mN'(\rho_A)=\sum_y \mathcal{U}^y_{B'''}\circ\Tr_{B''B}\left[\{|\Phi^+_{B'''B''}\>\<\Phi^+_{B'''B''}|\otimes\mN(\rho_A)\}\ \{\openone_{B'''}\otimes P^y_{B''B}\}  \right],
	\]
	for all $\rho_A\in\sD(\sH_A)$. The above identity can be equivalently written as the channel identity $\mN'=\mC\circ\mN$, upon introducing the map $\mC:\sL(\sH_B)\to\sL(\sH_{B'})$ defined as
	\[
	\mC(\bullet_B):=\sum_y \mathcal{U}^y_{B'''}\circ\Tr_{B''B}\left[\{|\Phi^+_{B'''B''}\>\<\Phi^+_{B'''B''}|\otimes\bullet_B\}\ \{\openone_{B'''}\otimes P^y_{B''B}\}  \right].
	\]
	\end{widetext}
	The above equation shows, in particular, that the map $\mC$ is a CPTP map defined everywhere, as claimed in the statement.
\end{proof}

A direct consequence of Proposition~\ref{prop:quantum} is that any quantum dynamical mapping that is not CPTP divisible gives rise to observable memory effects, in the following sense:

\begin{corollary}\label{cor:memory-effects}
	A quantum dynamical mapping $(\mN^i)_{i\ge0}$ is not divisible into linear CPTP maps, if and only if there exists an auxiliary Hilbert space $\sH_{S'}$, a finite ensemble of bipartite states $\widetilde{\mE}=\{p(x);\rho^x_{SS'}\}_x$, and a time $t_{\bar k}$ such that
	\[
	\pg(\widetilde{\mE}_{\bar k})>\pg(\widetilde{\mE}_{\bar k-1}).
	\]
\end{corollary}

In other words, it is possible to observe, at some time-step $t_{\bar k}$, a strict increase in the distinguishability for an initial ensemble $\widetilde{\mE}$. This can only happen if some memory is being kept during the evolution.

\section{Discussion}\label{sec:discussion}

Propositions~\ref{prop:classical} and~\ref{prop:quantum} above establish that divisibility of a discrete-time dynamical mapping is equivalent to a monotonic decrease of information (as measured by the guessing probability). This in particular implies (see Corollary~\ref{cor:memory-effects} above) that, as soon as a discrete-time dynamical mapping is \textit{not} divisible, then there necessarily exists an initial ensemble of quantum states whose guessing probability \textit{strictly} increases at some point along the evolution. Propositions~\ref{prop:classical} and~\ref{prop:quantum} hence provide the information-theoretic underpinning of divisibility, which therefore constitutes the key feature of memoryless processes. This equivalence is also valid in the case of continuous-time stochastic processes, since the latter can be obtained from the discrete-time setting by considering instants in time which are arbitrarily close to each other.

In this respect, our approach can be seen as a generalization of the idea first proposed by Breuer, Laine, and Piilo in~\cite{Breuer2009}. They characterize stochastic processes by tracking the change in the distinguishability of two different initial states of the system under dynamic evolution. However, while in~\cite{Breuer2009} only equiprobable pairs of states are considered, here we track the evolution of arbitrary ensembles of quantum states, i.e., ensembles consisting of more than two states in general, with arbitrary \textit{a priori} probabilities, and possibly living on an extended Hilbert space.

Our condition is therefore stronger than that in~\cite{Breuer2009}, and it is indeed equivalent to divisibility, while the criterion proposed in~\cite{Breuer2009} is only necessary but not sufficient, as explicitly shown by Chru\'{s}ci\'{n}ski, Kossakowski, and Rivas in~\cite{Chruscinski2011}. In fact, building upon the results of~\cite{kossakowski1972necessary},  Chru\'{s}ci\'{n}ski \textit{et al.} also propose a strengthened version of the criterion of Breuer \textit{et al.}, which is similar to ours, though based on a completely different proof strategy.

In particular, their criterion can only be applied to sequences of quantum channels $(\mN^i)_{i\ge 0}$ that are all bijective (as linear maps), thus excluding physically relevant situations (for example, semiclassical processes such as quantum measurements and decoherence, but not only those~\cite{nonsing,maldonado}) and common information-theoretic processes (for example, encodings, decodings, measurements, etc) that are typically non bijective.

The assumption of bijectivity is very limiting not only in quantum information theory, but also in classical information theory, in which Markov chains are typically used to model encoding-channel-decoding schemes~\cite{cover2012elements}. As such, it is clear that the assumption of bijectivity eludes a purely information-theoretic or operational description and must be put `by hand' on top of the dynamical evolution. Hence, a merit of our approach is that it does not require \textit{any} assumption about the underlying stochastic process: the channels constituting the dynamical mapping can be completely arbitrary, and our results can therefore be applied to any possible situation.

Note that Proposition~\ref{prop:quantum} also provides an operational characterization of reversible stochastic processes, as those for which the guessing probability is constant, i.e.,
\[
\pg(\widetilde{\mE}_i)=\pg(\widetilde{\mE}_{i+1}),\ 0\le i\le N-1,
\]
for any initial ensemble. The above equality implies not only the existence of `direct propagators', i.e., quantum channels $\mC^i$ such that $\mC^{i+1}\circ\mN^i=\mN^{i+1}$, but also the existence of `reverse propagators,' i.e., quantum channels $\mR^i$ such that $\mR^{i-1}\circ\mN^i=\mN^{i-1}$ (recall the consistency requirement $\mN^0=\id$). In other words, a stochastic process which preserves information has to be \textit{reversible}. In the particular case in which the system's Hilbert space remains the same (i.e., $\sH_i\cong\sH_S$ for all $i\ge1$), since the only reversible CPTP maps from $\bound(\sH)$ into itself 
are unitary ones~\cite{footnote1}, we arrive at the following conclusion: \textit{the only dynamical mappings which preserve information perfectly 
are those describing the evolution of a closed system}. This is in keeping with intuition, since a closed system has no environment to serve as memory: its evolution is, therefore, automatically Markovian~\cite{footnote2}.

A further (perhaps surprising) observation is that, according to the results in~\cite{buscemi2012comparison,Buscemi2014}, the ensembles of bipartite states $\widetilde{\mE}$ used in Definition~\ref{def:ext-info-dec} can without loss of generality be restricted to ensembles of \textit{separable states}.  This is because, as shown in Ref.~\cite{Buscemi2014}, the identity channel $\id_{S'}$, used in Definition~\ref{def:ext-info-dec} to define the partial ordering relation $\{\id_{S'}\otimes\mN^i\}\succeq\{\id_{S'}\otimes\mN^{i+1}\}$, can be replaced, without loss of generality, with some other noisy channel $\mM_{S'}:\sL(\sH_{S'})\to\sL(\sH_{S'})$, under the sole condition that $\mM_{S'}$ is \textit{complete}~\cite{Buscemi2014}, i.e., its image spans the whole $\sL(\sH_{S'})$. Since there exist complete quantum channels which are entanglement-breaking (e.g., a depolarizing channel $\mD^\varepsilon(\omega)=\varepsilon\omega+(1-\varepsilon)d^{-1}\openone$ with sufficiently small but nonzero $\varepsilon$), we arrive at the conclusion that in Definition~\ref{def:ext-info-dec} it actually suffices to consider bipartite states $\rho^x_{S'S}$ which are separable. In principle this fact may simplify the experimental assessment of Markovianity, since entanglement is not needed.

The authors acknowledge helpful discussions with Jeremy Butterfiled, Weien Chen, Dariusz Chru\'{s}ci\'{n}ski, David Reeb, and Sergii Strelchuk. F.B. acknowledges financial support from the JSPS KAKENHI, No. 26247016, and hospitality from the Statistical Laboratory of the University of Cambridge, where this research was done.

\appendix


\begin{thebibliography}{99}
\bibitem{footnote2} An equivalent, more formal way to express the same idea is that every group is, in particular, a semigroup.
\end{thebibliography}
\end{document}